\documentclass[conference]{IEEEtran}

\usepackage{amsfonts}
\usepackage{dsfont}
\usepackage{setspace}
\usepackage{color}
\usepackage{amssymb}
\usepackage{cite}
\usepackage[cmex10]{amsmath}
\usepackage{algorithm}
\usepackage{algorithmic} 
\usepackage{array}
\usepackage{mathrsfs}  
\usepackage{graphicx}
\usepackage{latexsym}
\usepackage{bbm}
\usepackage{amscd}
\usepackage{amsfonts}
\usepackage{subfigure}
\usepackage{amsmath,amscd,amssymb,verbatim}
\usepackage{graphics}
\usepackage{amsthm}
\usepackage[utf8]{inputenc}
\usepackage{authblk}
\usepackage{amsmath,amsthm}
\newtheorem{proposition}{\underline{Proposition}}
\newtheorem{lemma}{\underline{Lemma}}
  
\IEEEoverridecommandlockouts
\begin{document}

\title{Optimal 1D Trajectory Design for  UAV-Enabled Multiuser Wireless Power Transfer    \thanks{*J. Xu is the corresponding author. }}
 
\vspace{-.2cm}
\author[$\dag$]{ 
Yulin Hu$^1$, Xiaopeng Yuan$^1$, Jie Xu$^{2*}$,  and Anke Schmeink$^1$\\
 $^1$ISEK Research Group, RWTH Aachen University, 52062 Aachen, Germany.  \\
 $^2$School of Information Engineering, Guangdong University of Technology,
 Guangzhou 510006, China. \\
  E-mail:   \{hu,~xiaopeng.yuan\}@ti.rwth-aachen.de,~jiexu@gdut.edu.cn,~schmeink@ti.rwth-aachen.de    \vspace{-.2cm}
}

\maketitle

\begin{abstract}
In this paper, we study an unmanned aerial vehicle (UAV)-enabled wireless power transfer (WPT) network, where a UAV flies at a constant altitude in the sky  to provide wireless energy supply for a set of  ground nodes with a linear topology.   Our objective is to maximize  the minimum received energy among all ground nodes by optimizing the UAV's one-dimensional (1D) trajectory, subject to the maximum UAV flying speed constraint. Different from previous works that only provided heuristic and locally optimal solutions, this paper is  the first  work to present  the globally  optimal 1D UAV trajectory solution  to the considered   min-energy maximization problem.   Towards this end, we first show that     for any given speed-constrained   UAV trajectory, we can always construct   a maximum-speed trajectory and a speed-free trajectory, such that their combination can achieve the same received energy at all these ground nodes. 
Next, we transform the original UAV-speed-constrained trajectory optimization  problem into an equivalent UAV-speed-free problem, which is then optimally solved via    Lagrange dual method.  The obtained  optimal 1D UAV trajectory solution follows the so-called  successive hover-and-fly (SHF) structure, i.e.,   the UAV  successively hovers at a finite number of hovering points each for an optimized hovering duration, and  flies among these  hovering points at the maximum speed.   Numerical results show that our proposed   optimal solution significantly  outperforms the benchmark  schemes in prior works under different scenarios. 

\end{abstract}
 \vspace{.05cm}
\begin{IEEEkeywords}
Unmanned aerial vehicle (UAV), wireless power transfer (WPT), energy fairness, trajectory optimization, successive hover-and-fly (SHF). 
\end{IEEEkeywords}

    \vspace{-.1cm}
\section{Introduction}     \vspace{-.02cm}
 In recent years, unmanned aerial vehicle (UAV)-enabled wireless applications have attracted increasing attentions, as UAVs can be used for   the quick deployment of  on-demand wireless systems. Thanks to the presence of line-of-sight (LoS) aerial-to-ground  wireless links between UAVs and ground devices,  UAV-enabled wireless networks are likely to have a better system performance than conventional terrestrial wireless networks~\cite{magazine_1}.
In general, there are two types of UAV-enabled wireless applications, i.e., wireless communication~\cite{Survey_1,Throughput_1,EnergyEfficient_1,Offloading_1,location_1,location_2,location3D_1,location3D_2} and wireless power transfer (WPT)~\cite{Xu_Globecom,Xu_TCOM,Xu_Iot,UAV_WPT_EC}.

In UAV-enabled wireless communication, the fully mobile  UAVs provide new degrees of freedom in improving the wireless performance via optimizing UAVs' quasi-stationary  deployment locations or time-varying locations over time (a.k.a. trajectories) \cite{Survey_1}. For instance, the prior works \cite{location_1,location_2,location3D_1,location3D_2,Offloading_1} considered   UAV-enabled cellular base stations (BSs), where the UAV's deployment   locations are  optimized   to provide  the maximum coverage for ground users \cite{location_1,location_2,location3D_1,location3D_2}, and to enhance  the performance of cell-edge users via data offloading \cite{Offloading_1}. In addition, in UAV-enable mobile relaying systems, the UAV trajectory   is jointly optimized    with the wireless resource allocation, so as    to maximize the throughput \cite{Throughput_1,Throughput_2} or  the energy efficiency~\cite{EnergyEfficient_1}.


On the other hand, motivated by the great success of integrating WPT into wireless networks \cite{EH_1,EH_2,EH_3}, UAV-enabled WPT has recently emerged as a promising solution to prolong the lifetime of low-power  sensors and IoT devices, by using UAVs as mobile energy transmitters (ETs) to power these devices \cite{Xu_Globecom,Xu_TCOM,Xu_Iot,UAV_WPT_EC}. In particular, by considering the UAV flying at a
fixed altitude, the works~\cite{Xu_Globecom,Xu_TCOM}   optimized the one-dimensional (1D) or two-dimensional (2D) UAV trajectory   to maximize the energy transfer performance for a UAV-enabled WPT network, subject to a maximum UAV speed constraints. In a two-user scenario in a linear topology, the authors in \cite{Xu_Globecom}   optimized the 1D UAV trajectory to  characterize the Pareto boundary of the achievable  energy region by the two users. This result is then extended to the general multiuser scenario in a 2D topology  in \cite{Xu_TCOM}, where the 2D UAV trajectory is optimized to maximize  the minimum received energy among these users. 
It is worth noting that the above approaches in  \cite{Xu_Globecom,Xu_TCOM} can only obtain the globally optimal solution in the extreme case with the UAV maximum speed constraints being ignored. For the general case with the UAV maximum speed constraints involved, these approaches can only obtain   heuristic and locally optimal solutions. 
To our best knowledge, for the UAV-enabled WPT networks, how to obtain  the optimal UAV trajectory solution and reveal its     structure  still remains unknown, even for the basic  case with two users in a linear topology. This thus motivates our investigation in this paper to provide an optimal 1D UAV trajectory design and  to characterize the structure of the optimal UAV trajectory.

In this paper, we consider a   UAV-enabled WPT network with a linear topology, where multiple ground nodes are deployed in a straight line, e.g., along with a river, road or tunnel.
 To charge these ground nodes in an efficient and fair manner, we  aim at maximizing the minimal received energy among all ground nodes via designing the UAV's  1D trajectory (or equivalently the velocity) for WPT, while the UAV mobility is subject to  maximum speed constraints. 
 The results of this work are summarized as follows:  Different from previous works that only provided heuristic and locally optimal solutions, for the first time, we present  the globally  optimal 1D UAV trajectory solution  to the considered WPT problem, by equivalently decomposing any speed-constrained 1D UAV trajectory into a maximum-speed trajectory and a speed-free trajectory, together with the Lagrange dual method. It is proved that the optimal 1D UAV trajectory solution follows an  interesting {\it successive hover-and-fly} (SHF)   structure, i.e.,   the UAV  successively hovers at a finite number of hovering points each for an optimized hovering duration, and  flies among these hovering points at the maximum speed. 

   \vspace{-.3cm}
 \section{System Model and Problem Formulation}
 \label{sec:model_and_Problem}
In this paper, we  consider a   UAV-enabled multiuser WPT system with
a linear topology as shown in Fig.~\ref{system_topology},  where a UAV flies at a fixed altitude $H>0$ to wirelessly charge a set $\mathcal K = \{1,\cdots,K\}$ of $K$   ground nodes (such as IoT devices and sensors) that
are located  in a straight line. 
\begin{figure}[!h]
\centering
\includegraphics[width=0.42\textwidth, trim=10 17 15 22]{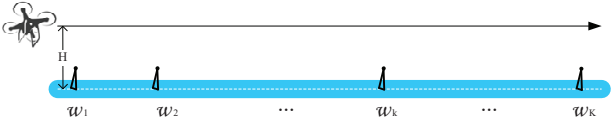}
\caption{Illustration of the UAV-enabled WPT network with a linear topology. }
\label{system_topology}
   \vspace{-.16cm}
\end{figure}
We denote the horizontal location of node $k \in \mathcal K$ as $w_k$. We assume that $w_1 \le \cdots \le w_K$ without loss of generality.
 To efficiently charge  all nodes, we focus on a finite UAV charging period     $\mathcal{T} \buildrel \Delta \over = [0,T]$ with duration $T > 0$.
The UAV's time-varying  horizontal location     is denoted by $x(t)$ at time instant  $t \in \mathcal T$.
In addition, the UAV is subject to a maximal flying speed   $V$. Hence, we have  $| \dot{x} (t)| \leq V, \forall t \in \mathcal{T}$, where $\dot{x}(t)$ denotes the first-order derivative of $x(t)$.

In practice, the  wireless channels between the UAV and   ground nodes  are  LoS-dominant, and therefore, we adopt the free-space path loss model as normally used in the UAV-enabled wireless communication and WPT literature~\cite{Throughput_1}.
 At  time $t$, the channel power gain from the UAV to ground node $k\in\mathcal K$ is denoted as $h_k(x(t)) = \frac {  \beta_0}{\left(x(t)-w_k\right)^2+H^2}$, where
 the distance between  the UAV and ground node $k$ is   $\sqrt{(x(t)-w_k)^2+H^{2}}$ and $\beta_0$ is the channel power gain at a reference distance of unit meter.
Hence,  the    received radio frequency (RF) power  by ground node $k$ at  time $t\in\mathcal T$~is 
\begin{equation}
Q_k\left(x(t)\right)= \frac {  \beta_0 P}{\left(x(t)-w_k\right)^2+H^2},
\end{equation}
where $P$ denotes the constant transmit power of the UAV. Notice that in practice, the received radio frequency (RF) signal should be converted into a direct current (DC) signal to charge the rechargeable battery at each ground node, and the RF-to-DC conversion is in general a non-linear process~\cite{Bruno_2018}. In order to focus our study on the wireless transmission, we use the received RF power as the performance metric by ignoring the non-linear RF-to-DC conversion process, as in ~\cite{Xu_Globecom, Xu_TCOM}.

Due to the broadcast nature of the wireless transmission,  all  ground nodes can simultaneously receive wireless  power during the whole  charging period $\mathcal{T}$.
As a result, the total energy received by   ground node $k \in\mathcal K$  is given by 
\begin{equation}
 \label{eq:device_k_energy}
E_k(\{x(t)\})
= \int_{0}^{T} Q_k\left(x\left(t\right)\right)dt.
\end{equation}



Our objective is to design the UAV trajectory  to maximize the minimal received energy among all the $K$ nodes  during the   charging period $T$.  The problem   is formulated as  
\begin{eqnarray}
 \label{eq:OP}
({\rm OP}): \max_{\{x(t)\}} &&  \min\limits_{k\in\mathcal K} \int_{0}^T Q_k(x(t))dt  \\
\mathrm{s.t.} 
 && | \dot{x} (t)| \leq V, \forall t \in \mathcal{T}. \nonumber
\end{eqnarray}
By introducing an auxiliary variable   $E$, 
the original  problem   (OP) is equivalently reformulated as  
\begin{eqnarray}
 \label{eq:P1}
({\rm P}1): \max_{\{x(t)\},E} && E  \nonumber \\
\mathrm{s.t.} && \int_{0}^{T} Q_k(x(t))dt\geq   E  ,\forall k \in \mathcal K \,  \label{eq:P1_constraint}  \\
 && | \dot{x} (t)| \leq V, \forall t \in \mathcal{T} \nonumber.
\end{eqnarray}

Notice that both the original problem (OP) and the     reformulated problem (P1) are non-convex, due to the fact that the objective function in (OP) is non-concave, and constraint ${\int_{0}^{T} Q_k(x(t))dt\geq E}$ in (P1) is non-convex, respectively. Furthermore, both problems consist of an infinite number of variables $\{x(t)\}$ over continuous time. Therefore, how to find the optimal solution to the min-energy maximization problem is generally a very difficult task.

Notice that in the prior work~\cite{Xu_TCOM}, the authors solved the 2D UAV trajectory optimization problem for min-energy maximization by the following  three-step approach, which can also be used to solve (OP) and (P1) directly. First, by considering a relaxed problem of (P1) with    the UAV's maximum speed
constraint ignored, the optimal {\it multi-location-hovering} trajectory solution of the relaxed problem is obtained. Next, by taking into account the maximum  UAV speed constraint,   a heuristic   SHF  trajectory design is proposed, in which the UAV flies at the maximum speed to successively visit the obtained optimal hovering locations to the relaxed problem above, and hovers above them accordingly. 
In the heuristic SHF trajectory, the traveling salesmen problem (TSP) is used to obtain the visiting order among these locations with minimized flying distance/time\footnote{Note that TSP is only needed for the 2D trajectory, but is not required for the 1D trajectory design of our interest. Nevertheless, the heuristic SHF is suboptimal as it does not take into account the WPT during UAV flying. }. Finally, the    {\it successive convex approximation (SCA)-based} trajectory design is proposed, which quantizes the path or time to subsequently refine the trajectory towards a locally optimal solution. 
It is worth noting that both the heuristic SHF and the SCP based    approaches can only obtain the globally optimal solution the relaxed problem, which mathematically corresponds to the ideal case with the flying duration or the UAV flying speed being infinite. This case may not happen in practice. For the general case, the heuristic SHF  trajectory is   suboptimal, while the performance of the SCP-based trajectory can only ensure the local optimality when the quantization becomes extremely accurate.  How to characterize the optimal 1D UAV trajectory    solution to the min-energy maximization problem in the general case with speed constraints   is still unknown.

    \vspace{-.24cm}
\section{Optimal SHF Trajectory Solution}    
 \label{sec:Optimality}
In this section, we present the optimal trajectory solution  to  problem (OP) or  (P1), and show that it has an interesting SHF     structure, in which the UAV hovers among a number of locations and then flies among them at the maximum speed.  


First, notice that there always exists a uni-directional trajectory that is optimal for problem (P1), i.e., $x(t_1) \le x(t_2), \forall t_1, t_2 \in \mathcal T, t_1 < t_2.$ This is due to the fact that for any given trajectory,  we can always find an alternative uni-directional UAV trajectory to achieve the same WPT performance but without flying forward and backward~\cite{Xu_JSAC_2018}. Therefore, in this paper we focus on the uni-directional trajectory without loss of optimality.

Next, we consider problem (P1) under given pair of initial and final locations $(x_{\rm I}, x_{\rm F})$.  This sub-problem is expressed  as 
\begin{eqnarray}
 \label{eq:P1.1}
({\rm P}1.1): \max_{\{x(t)\},E} && E  \nonumber \\
\mathrm{s.t.} && \int_{0}^{T} Q_k(x(t))dt\geq   E  ,\forall k \in \mathcal K \,   \label{eq:P11_constraint}  \\
 && | \dot{x} (t)| \leq V, \forall t \in \mathcal{T} \nonumber  \\
  && x_{\rm I} \leq  {x}(t) \leq x_{\rm F}, \forall t\in  {\mathcal T} .  \nonumber    \vspace{-.2cm}
\end{eqnarray}
In the following, we first show that any speed-constrained trajectory to problem (P1.1) 
is mathematically equivalent to the combination of a maximum-speed trajectory and a speed-free trajectory, and then provide the optimal solution to problem (P1.1)  under any  given $x_{\rm I}$ and $x_{\rm F}$
 via the Lagrange dual method. 
 After that,  we obtain the optimal solution to problem  (P1) by applying a 2D exhaustive search over $x_{\rm I}$, $x_{\rm F} \in [w_1,w_K], x_I \le x_{\rm F}$. 
 

\subsection{Problem Reformulation}
We start with the following lemma to show that we can construct two trajectories for any unidirectional trajectory $\{x(t)\}$ satisfying the maximum speed $V$.   



\begin{lemma}
 \label{le:Lemma1}
For any duration-$T$ unidirectional trajectory $\{x(t)\}$ satisfying the maximum speed $V$ with given initial position $x(0) = x_{\rm I}$ and final position $x(T) = x_{\rm F}$,  we can always find two UAV trajectories $\{\bar x(t)\}$   and $\{\hat x(t)\}$   to jointly achieve the same WPT performance. In particular, $\{\bar x(t)\}$ is the max-speed flying with $\bar x(t) = x_{\rm I} + Vt , \forall t\in(0,\bar T]$, where   $\bar T=(x_{\rm F}-x_{\rm I})/V$.
In addition, $\{\hat x(t)\}$ has a time duration $\hat T=T-(x_{\rm F}-x_{\rm I})/V$ without any UAV speed constraints (speed-free). In other words, the following equality holds for any $ k \in \mathcal K$.    \vspace{-.1cm}
\begin{equation}
 \int_{0}^{T} \! Q_k\left(x\left(t\right)\right)dt =  \! \int_{0}^{\bar T} \! Q_k\left(\bar x\left(t\right)\right)dt +  \int_{0}^{\hat T} \! Q_k\left(\hat x\left(t\right)\right)dt.
 \end{equation}
\end{lemma}
\begin{proof}
The proof is provided in   Appendix~A.    \vspace{-.2cm}
\end{proof}

Note that for the maximum-speed flying trajectory $\{\bar x(t)\}$  from $x_{\rm I}$   to  $x_{\rm F}$, the trajectory   is fixed, i.e., the UAV flies from  $x_{\rm I}$   to  $x_{\rm F}$ at the maximal speed $V$.
In particular, for   trajectory $\{\bar x(t)\}$, the received energy by ground node $k\in\mathcal K$ is       
\begin{equation}
\begin{split}
 \label{bar_x_energy_k}
\! \!\bar E_k &=  \int_{0}^{\bar T} \!\! Q_k \! \left(\bar x\left(t\right)\right)dt \! =\! \int_{0}^{\bar T} \!\!Q_k\left(  x_{\rm I} + Vt \right)dt\\ 
&= \frac{\beta_0 P}{VH}\arctan(\frac{x_{\rm F}\!-\!w_k}{H})\!-\! \frac{\beta_0 P}{VH}\arctan(\frac{x_{\rm I}\!-\!w_k}{H}).
 \end{split}
 \end{equation}

%

Based on  Lemma~\ref{le:Lemma1},  problem (P1) under given $x_{\rm I}$   and $x_{\rm F}$  can be equivalently reformulated~as   
\begin{eqnarray}
 \label{eq:P2}
\!\!\!\!\!\!\!\!({\rm P}2): \max_{\{\hat x(t)\},E} && E \nonumber \\
\mathrm{s.t.} && \int_{0}^{\hat T} Q_k(\hat{x}(t))dt+ \bar{E}_k \geq   E  , \, \, \forall k \in \mathcal K  \,   \\
 && x_{\rm I} \leq \hat{x}(t) \leq x_{\rm F}, \forall t\in \hat{\mathcal T},  \nonumber 
\end{eqnarray} 
 where $\hat {\mathcal T} \triangleq    [0,\hat{T}]$
and $ \bar{E}_k$ is a constant defined in~\eqref{bar_x_energy_k}.


\subsection{Optimal Solution to Problem (P2)}    
 \label{sec:Optimality_P2}

 Problem (P2) is non-convex but satisfies the so-called time-sharing condition in \cite{Yu_2006}. Therefore,   strong duality holds between   problem (P2) and its Lagrange dual problem. Therefore,    problem (P2) can be solved  via the    Lagrange dual method  \cite{S_2004}.
%

Denote the  Lagrange multiplier for the $k$-th constraint in~(8)   by  $\lambda_k\geq 0, k\in \mathcal K$. 
The partial Lagrangian of problem (P2) is 
\begin{equation}
\begin{split}
\!\!&\mathcal{L}_2 \left(\{\hat x(t)\},E,\{\lambda_k\} \right) = E\!+\!\sum_{k\in \mathcal K}\!\lambda_k\left(\!\int_{0}^{\hat T} \!\!Q_k(\hat x(t))dt \!+\! \bar{E}_k\!-\!E\right) \\
\!\!&= (1-\sum_{k\in \mathcal K}\lambda_k)E+ \sum_{k\in \mathcal K}\lambda_k \bar{E}_k +\int_{0}^{\hat T}\sum_{k\in \mathcal K}\lambda_k Q_k(\hat x(t)) dt.
\end{split}
\end{equation}
 %
Immediately, we have the corresponding dual function as 
\begin{eqnarray}
  \label{eq:dual_P}
f_2\left(\{\lambda_k\}\}\right)&= &\max_{\{\hat x(t)\},E} \mathcal{L}_2\left(\{\hat x(t)\},E,\{\lambda_k\}\}\right)  \\
&\mathrm{s.t.}  & x_{\rm I} \le \hat{x}(t) \le x_{\rm F}, \forall t\in\hat{\mathcal T} \nonumber. 
\end{eqnarray}
Clearly, the condition $1-\sum_{k\in \mathcal K}{\lambda_k}=0$ must be satisfied to guarantee that the function $f_2(\{\lambda_k\})$ is upper-bounded from above, i.e., $f_2(\{\lambda_k\}) < \infty$. Otherwise, if $1-\sum_{k\in \mathcal K}{\lambda_k} < 0$ (or $1-\sum_{k\in \mathcal K}{\lambda_k} > 0$), we have $f_2\left(\{\lambda_k\}\}\right) \to \infty$ by setting  $E \to - \infty$ (or $E \to \infty$). 
Then, the dual problem of   problem (P2) is given by
\begin{eqnarray}
({\text{DP2}}): \max_{\{\lambda_k\}} && f_2(\{\lambda_k\})   \\
\mathrm{s.t.} && 1-\sum_{k\in \mathcal K} {\lambda_k} =0,\,\,\,\,  \lambda_k\geq 0 , \forall k \in \mathcal K \nonumber .
\end{eqnarray}

\begin{table*}[!h]
    \vspace{.036in}
\begin{equation}
\label{x_nex_expression}
\!x^*(t)\!=\!\!
\left\{\!\!\! \begin{array}{llr}
\hat x_i^*,
& \textrm{if} \, \sum\limits_{j=1}^{i}\hat\tau^*_j \! -\! \hat\tau^*_i  \!+\!\frac{x^*_i\!-\!x^*_0}{V} \leq t \leq \sum\limits_{j=1}^i \hat\tau^*_j\!+\!\frac{x^*_i\!-\!x^*_1}{V}, \, \,  {1\leq i \leq N+1},\\
x^*_0+V(t- \sum\limits_{j=1}^i \hat\tau^*_j),  & \textrm{if} \,  \sum\limits_{j=1}^i \hat\tau^*_j\!+\!\frac{x^*_i\!-\!x^*_0}{V} \leq t \leq \sum\limits_{j=1}^i \hat\tau^*_j\!+\!\frac{x^*_{i+1}\!-\!x^*_0}{V},  \, \, {1\leq i \leq N}, \\
\end{array} \right.
\end{equation}
\hrule 
\end{table*}

 Notice that as strong duality holds between problem (P2) and its dual problem (DP2), we solve problem (P2) by equivalently solving the dual problem (DP2), in which we first obtain the dual function $f_2\left(\{\lambda_k\}\}\right)$ under any given $\{\lambda_k\}$ by solving problem \eqref{eq:dual_P}, and then updating $\{\lambda_k\}$ via subgradient-based methods such as the ellipsoid method~\cite{Ellipsoid_Method} to find the opitmal  $\{\lambda_k^\star\}$ maximizing $f_2(\{\lambda_k\})$.  

First of all, we deteminate the dual function $f(\{\lambda_k\})$. 
Consider problem \eqref{eq:dual_P} under any given~$\{\lambda_k\}$ satisfying the constraints in (DP2).
As $1-\sum_{k\in \mathcal K}{\lambda_k}=0$, problem (10) can be decomposed into the following subproblems by dropping the constant term $\sum \lambda_k \bar E_k$,  each for one time instant $t \in \mathcal T$.    
 \begin{eqnarray}
  \label{eq:dual_or_P2_}
\max_{\{\hat x(t)\}} && \sum_{k\in \mathcal K}\lambda_kQ_k(\hat x(t))  \\
\mathrm{s.t.} && x_{\rm I}\leq \hat x(t) \leq x_{\rm F}. \nonumber
\end{eqnarray}
As the problem in \eqref{eq:dual_or_P2_} has the same form 
at different  time instant $t$ in the above problem, we can simply drop the variable $t$ and   re-express the   problem  as 
\begin{eqnarray}
 \label{eq:dual_wot1_}
\max_{x} &&  F(\hat x)\triangleq  \sum_{k\in \mathcal K}\lambda_kQ_k(\hat x). \\
\mathrm{s.t.} && x_{\rm I} \leq \hat x \leq x_{\rm F}. \nonumber
\end{eqnarray}
We obtain the  extreme points of $F(\hat x)$ by letting  its first order derivative be zero, i.e.,  
\begin{equation}
F'(\hat x)=\sum_{k\in \mathcal K}\lambda_k\frac {-2(\hat x-w_k)\beta_0 P}{(\left(\hat x-w_k\right)^2+H^2)^2}=0,
\end{equation}
 which equals to solve    \vspace{-.25cm}
\begin{equation}
 \label{eq:dual_wot11}
\sum\limits_{k = 1}^K \left\{ - 2(\hat x - {w_k}){\beta _0}P{\lambda _k}  \cdot \prod\limits_{i \ne k}^{i \in \mathcal K} {\left({{\left( {\hat x - {w_i}} \right)}^2} + {H^2}\right)^2} \right\} =0.
\end{equation}
In other words,  the  extreme points of $F(\hat x)$ can be obtained by solving~\eqref{eq:dual_wot11}. By comparing the objective values in \eqref{eq:dual_wot1_} at the extreme points versus those     at  the boundary points $x_{\rm I} $ and $ x_{\rm F}$,     the optimal hovering points  $\hat x_1^*$, $\hat x_2^*, ..., \hat x_N^*$ are obtained, where       $N$ denotes the number of optimal solutions which achieves the same objective value. Accordingly, the dual function $f(\{\lambda_k\})$  is obtained.

With $f(\{\lambda_k\})$ obtained, the dual problem (DP2)  can be solved via the ellipsoid method,   and therefore,   the solution $\{\lambda_k^*\}$ is obtained. Based on $\{\lambda_k^*\}$, we  can  reconstruct the primal optimal solution to (P2) by solving    the time-sharing problem for allocating   the total duration  $T$ over  the~$N$ hovering points,  for which the optimization problem is formulated as the following linear program (LP):
\begin{eqnarray}
 \label{eq:dual_wot1_}
\max_{\{\hat \tau_i \ge 0\}, \hat E} && \hat  E \nonumber \\
\mathrm{s.t.} && \sum\nolimits^N_{i=1} \hat \tau_i Q_k(\hat x_i) \ge \hat E., \forall    {k\in \mathcal K}  \\
 &&\sum\nolimits^N_{i=1} \hat \tau_i  =T. \nonumber
\end{eqnarray}
By solving this LP problem via standard interior point method, we  obtain  the  optimal hovering durations $\hat \tau _1^*$, $\hat\tau _2^*, ..., \hat\tau _N^*$ corresponding to the $N$ hovering points.
 Therefore, (P2) is optimally solved under the given $x_{\rm I}$ and $x_{\rm F}$.
 In summary, the     optimal solution to (P2) is  described by the optimal hovering points and hovering durations
 \begin{equation}
  \label{eq:opt_expression_P2}
 \hat x^*(t) = \hat x_i^*, \,\,  {\rm if}\,\,\,  t \in \Big [\sum\nolimits_{j=1}^i  \hat\tau^*_j \!-\! \hat\tau^*_i,\,\, \sum\nolimits_{j=1}^i \hat\tau^*_j\Big],
\end{equation}
where ${ i = 1,\cdots,N}$.      \vspace{-.15cm}
%
\begin{lemma}
\label{th:infinity-speed}
There exists     one optimal multi-location-hovering solution to problem (P2) with the number of hovering points being no more than $K$, i.e.,
$\hat N \leq 2K+1$.       \vspace{-.05cm}
\end{lemma}
\begin{proof}
 It is observed that the optimal solution to (P2) has a multi-location-hovering structure, i.e., the solution to problem \eqref{eq:dual_wot1_} is a set of hovering points. 
  Note that   the left side of \eqref{eq:dual_wot11} is a  $4K-3$ order  polynomial of $\hat x$, i.e., $F(\hat x)$ has at most $4K-3$ extrema and therefore at most  $2K-1$ maximum points as potential hovering points. In addtioin,  the two boundary points $x_{\rm I}$ and $x_{\rm F}$  are also   potential hovering points.  
  Hence, there are a maximum number of $2K+1$ hovering locations in    the optimal solution to (P2).    
\end{proof}

    \vspace{-.3cm}
\subsection{Optimal Solution to Problem (P1.1)}
 \label{sec:Opt_design}

In Section~\ref{sec:Optimality_P2}, we have shown that for given $x_{\rm I}$ and $x_{\rm F}$, the   global optimal trajectory problem (P2)  can be obtained.   Suppose that  the corresponding optimal solution to problem (P2)   $\hat x(t) $ with optimal hovering points  $\hat x_1^*$, $\hat x_2^*, ..., \hat x_N^*$   and   the corresponding optimal hovering durations.  $\hat\tau _1^*$, $\hat\tau _2^*, ..., \hat\tau _N^*$.
According to Lemma~1, we can express the optimal solution to (P1.1) by combining  $\hat x(t)$ with $\bar x(t) = x_{\rm I} + Vt , \forall t\in(0,\bar T]$,  i.e., letting the UAV fly at the maximal speed from $x_{\rm I}$ to $x_{\rm F}$ while stopping/hovering   at the $N$  hovering points (in between) with the corresponding optimal hovering durations. 
 By defining $x^*_0 = x_{\rm I}^*$, $x^*_{N+1} = x_{\rm F}^*$ and $\hat\tau _0^*= \hat\tau _{N+1}^*=0$, we have  the optimal   solution $\{x^*(t)\}$  to problem (P1.1), where   $x^*(t)$ is given in~\eqref{x_nex_expression}. 

    \vspace{-.1cm}
\subsection{Optimal Solution to (OP) or  (P1)}      \vspace{-.05cm}
 \label{sec:Opt_design}
In Section~\ref{sec:Opt_design}, we have shown that  the global optimal trajectory to  problem (P1.1) is obtained, i.e.,  we have optimally solved  problem
(P1)  under given $x_{\rm I}$ and $x_{\rm F}$. 
Hence, by applying a 2D exhaustive search over the possible pair of  $x_{\rm I}$ and $x_{\rm F}$ together with solving (P1.1) under each $x_{\rm I}$ and $x_{\rm F}$, 
the global optimally trajectory solution to problem (P1)  is finally obtained.
  It is clear there is no benefit if the UAV hovers at a position out of the region of  ground nodes. Hence,  the feasible set of $x_{\rm I}$ is $[w_1,w_K]$ while the corresponding feasible set of $x_{\rm F}$ is $[x_{\rm I},w_K]$. 
To apply the  exhaustive search on  $x_{\rm I}$ or $x_{\rm F}$ within its  continuous feasible set, we introduce $d_{\min}$ as the resolution in distance.
Note that  we have $d_{\min}$ should be small and let $\frac{w_K-w_1}{d_{\min}}$ be an integer. 
Hence, the feasible locations of $x_{\rm I}$ and the corresponding $x_{\rm F}$ become $\{w_1, w_1\!+ d_{\min},w_1\!+\!2 d_{\min}, \cdots, w_K \}$ and $\{x_{\rm I}, x_{\rm I}\!+\! d_{\min},x_{\rm I}\!+\!2 d_{\min}, \cdots, w_K \}$, respectively.   


The detail description of the optimal solution to (P1) is provided in Algorithm~1.   
   \vspace{-.11in}
\begin{algorithm}[!h]
\algsetup{linenosize=\large}
\vspace{.05in}
\caption{\bf{for Optimally Solving Problem (P1)}}
\begin{algorithmic}
  \STATE $\!\!\!\!\!r_1=0$, $r_2=0$.
 \STATE$\!\!\!\!\!\!\!$ \noindent{\bf{for $\!x_{\rm I} = w_1 :  d_{\min} : w_K$}} \\
 \STATE $\!\!\!\!\!\!\!$\noindent{~~~~~$r_1=r_1+1$}; ~$x^{(r_1)}_{\rm I}=w_1+ r_1 \cdot d_{\min}$.   \\
 \STATE$\!\!\!\!\!\!\!$ \noindent{~~~~\bf{for $\!x_{\rm F} = x^{(r_1)}_{\rm I} :  d_{\min} : w_K$}} \\
  \STATE $\!\!\!\!\!\!\!$\noindent{~~~~~~~~~~$\!\! r_2\!=\! r_2 \!+\!1$};~$x^{(r_2)}_{\rm F}= x_{\rm I}^{(r_1)} \!+\! r_2 \!\cdot\! d_{\min}$. Hence, we have a \\    \STATE~~~~~~~~~~~~~~~~~~~~~~~~~~~~~~problem (P1.1) with   $(x^{(r_1)}_{\rm I},x^{(r_2)}_{\rm F})$. \\
   \STATE$\!\!\!\!\!\!\!$ \noindent{\bf{~~~~~~~~$\!\!$Optimally Solving Problem (P1.1)}} \\
 \STATE $\!\!\!\!\!\!\!$\noindent~~~~~~~~~{\bf{a)}}~According to Lemma~1, obtain the corresponding     \\
\STATE~~~~~~~~~$\!\{\bar x^{(r_1,r_2)}(t)\!\}\!$ and problem (P2). 
 \STATE$\!\!\!\!\!\!\!$ \noindent~~~~~~~~{\bf{b)}}~Solve       (P2) according to Section~\ref{sec:Optimality_P2} and have  the \\
\STATE~~~~~~~~~optimal solution $\!\{\hat x^{(r_1,r_2)}(t)\}$. \\
 \STATE $\!\!\!\!\!\!\!$\noindent~~~~~~~~~{\bf{c)}}~Combining  $\{\bar x^{(r_1,r_2)}(t)\}$ with  $\{\hat x^{(r_1,r_2)}(t)\}$, we \\
\STATE~~~~~~~~~have  $\{x^{(r_1,r_2)}(t)\}$, which is the optimal  solution to 
\STATE~~~~~~~~~the problem (P1.1) with  $(x^{(r_1)}_{\rm I},x^{(r_2)}_{\rm F})$.
 \STATE $\!\!\!\!\!\!\!$\noindent{~~~~\bf{end}} \\
   \STATE $\!\!\!\!\!\!\!$\noindent{~~\bf{$\!\!\!\!$end}} \\
   \vspace{-.0051in}
   \STATE$\!\!\!\!\!\!\!$ \noindent{The optimal solution to problem (P1) is  $\{x^{*}(t)\} = \arg \max\limits_{r_1,r_2}    \Big\{  \min\limits_{k\in\mathcal K}  E_k (\{x^{(r_1,r_2)}(t)\})   \Big\}  $}. \\
\end{algorithmic}
\label{algorithm1}
\end{algorithm}


     \vspace{-.35cm}
\subsection{Structure of the Optimal Trajectory   to Problem (P1)}     \vspace{-.05cm}

We describe the      structure  of  the  optimal trajectory solution to problem  (P1) in the following proposition.      
\begin{proposition}
\label{th:structure1}
The optimal trajectory solution to problem (P1.1) or problem (P1) follows the   SHF   structure, i.e., there exists a number of $N$ hovering locations at the optimal trajectory, such that the UAV always flies at the maximum speed from one hovering location to another, and then hovers at that location for a certain time duration.  It then holds that
$N \leq 2K+1$.
\end{proposition}

\begin{proof}
Combining Lemmas 1 and 2, this proposition is verified directly for problem (P1.1).
Note that the global optimal trajectory to problem (P1)  is obtained  by applying a 2D exhaustive search over the possible pair of  $x_{\rm I}$ and $x_{\rm F}$, i.e., the optimal solution to problem (P1) is the best one in the solutions of all the  problems (P1.1) with different  $x_{\rm I}$ and $x_{\rm F}$.  Hence, as the optimal solution to any problem (P1.1) has a SHF structure,  the global optimal trajectory to problem (P1)    also has such a structure, i.e., Proposition~1 holds also for problem (P1).
\end{proof}

\section{Numerical Results}
\label{sec:Evaluation}
In this section, we evaluate the   proposed optimal SHF algorithm,    in comparison to two reference algorithm, i.e., the heuristic SHF and the SCP with time quantization from~\cite{Xu_TCOM}.
To obtain the WPT performance, we randomly drop ground nodes to have $20$ different topologies, and then apply these three algorithms at each topology, and finally   average   the max-min received  power among all ground nodes over these  random realizations.
In the simulation, we have the following default setups of parameters: $\beta_0=-30~dB$, $P=40~dBm$, $K=5$, $H=5~m$, $T=20~s$, $V=1~m/s$ and $\mathbf{w}$ is a $K$-dimension vector with each element being a random number in interval $[0,D]$, where $D=20~m$.  In addition, we set the quantization size of distance for the exhaustive search in the proposed algorithms to $d_{\min}=0.01~m$ and set accordingly the time quantization size in the reference algorithm SCP with time quantization to $t_{\min}=d_{\min}/V$. 
\begin{figure}[!t]
\label{resultvaryingT_average}
\centering
\includegraphics[width=0.44\textwidth, trim=10 15 15 15]{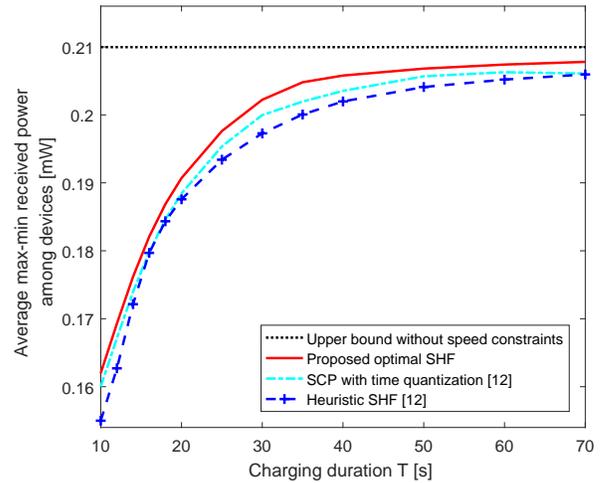}
\caption{Average performance comparison with a varying charging duration.}\label{comparison}
    \vspace{-.4cm}
\end{figure}
\begin{figure}[!t]
\centering
\includegraphics[width=0.45\textwidth, trim=2 15 17 0]{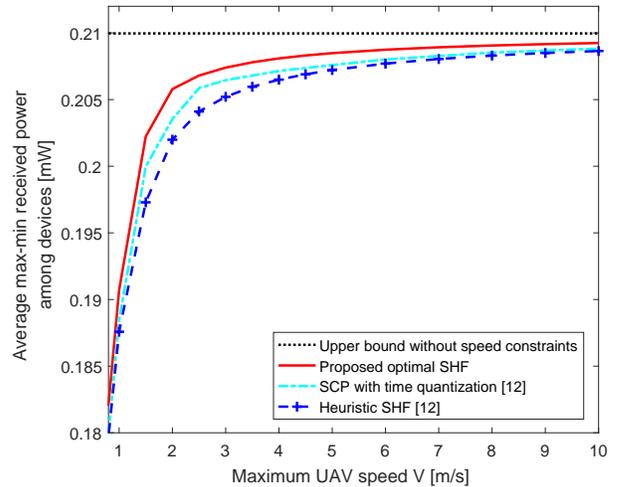}
\caption{Average performance comparison with a varying~speed constraint.}\label{comparison2}
    \vspace{-.40cm}
\end{figure}
 
 \begin{figure*}[!t]
\centering
\includegraphics[width=0.8\textwidth, trim=0 8 25 5]{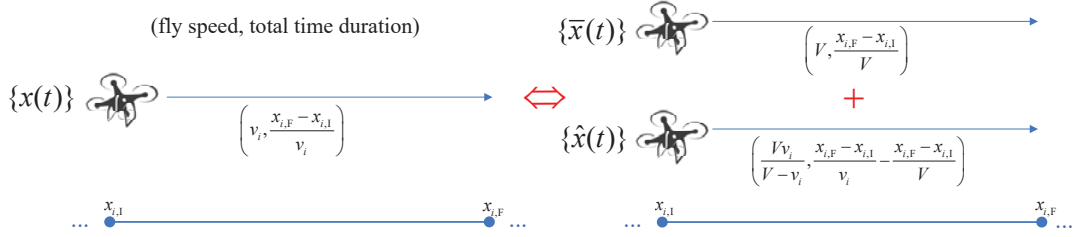}
\caption{Under Case 3, the WPT performance of $\{x(t)\} $ from $x_{i, \rm I}$ to $x_{i, \rm F}$   is  equivalent to   $\{\bar x(t)\} $ together with  $\{\hat x(t)\} $.}\label{comparison1}
 \label{Fig:forcase_3}
     \vspace{-.15cm}
\end{figure*}
\begin{table*}[!h]
\hrule
 \vspace{-0.1cm}
\begin{eqnarray}
\label{eqn_lemma1}
\!\!\!\!\!\!\!\!\!\!\!\!\!\!\!&\!\!\!& \underbrace{\int_0^{\frac{x_{i,\rm{F}}-x_{i,\rm{I}}}{v_i}}\frac{\beta_0 Pdt}{(x_{i,\rm{I}}+v_it-w_k)^2+H^2}}_{\{  x(t)\} } \nonumber \\
\!\!\!\!\!\!\!\!\!\!\!\!\!\!\!&\!\!\!=& \frac{v_i}{V}\int_0^{\frac{x_{i,\rm{F}}-x_{i,\rm{I}}}{v_i}}\frac{ \beta_0 Pdt}{(x_{i,\rm{I}}+v_it-w_k)^2+H^2}+\frac{V-v_i}{V}\int_0^{\frac{x_{i,\rm{F}}-x_{i,\rm{I}}}{v_i}}\frac{\beta_0 Pdt}{(x_{i,\rm{I}}+v_it-w_k)^2+H^2}    \\
\!\!\!\!\!\!\!\!\!\!\!\!\!\!\!&\!\!\!\buildrel {\buildrel v_it=Vt_1  \over {v_it=\frac{Vv_i}{V-v_i} t_2}} \over {=}& \underbrace{\int_0^{\frac{x_{i,\rm{F}}-x_{i,\rm{I}}}{V}}\frac{\beta_0 Pdt_1}{(x_{i,\rm{I}}+Vt_1-w_k)^2+H^2}}_{\{\bar x(t)\} }+\underbrace{\int_0^{\frac{x_{i,\rm{F}}-x_{i,\rm{I}}}{v_i}-\frac{x_{i,\rm{F}}-x_{i,\rm{I}}}{V}}\frac{\beta_0 Pdt_2}{(x_{i,\rm{I}}+\frac{Vv_i}{V-v_i} t_2-w_k)^2+H^2}}_{ {\{\hat x(t)\} }} . \nonumber
\end{eqnarray}
    \vspace{-.1cm}
\hrule
     \vspace{-.25cm}
\end{table*}
We  first study the impact of the charging duration $T$ on the max-min received power (among all ground nodes). The results are shown in Fig.~\ref{comparison}, where the upper bound of the ideal case with UAV speed constraint ignored is also provided.
First, it  is   observed       that, as the charging duration $T$ increases, the performance of all the three algorithms   increases towards the upper bound.
In addition, owning to applying additional SCP process, the SCP with time quantization has a better performance than the heuristic SHF, which is consistent with the results in~\cite{Xu_TCOM}.
More importantly,   the proposed     optimal SHF outperforms the two reference algorithms, in the whole charging duration regime. 

The relationship between the max-min received power (among all ground nodes) and UAV speed $V$ is investigated in Fig.~\ref{comparison2}.
From the figure, we learn that under all algorithms, the max-min received power   increases as the  UAV speed~$V$ becomes large.
In addition, as the speed significantly increases, all the three designs are observed to approach   the upper bound. 
Moreover, we observe again the performance advantage of the proposed optimal   algorithm  in comparison to the two reference algorithms.

\section{Conclusion}
\label{sec:Conclusion}

In this paper, we focus on a UAV-enabled multiuser WPT network with a linear topology.  We   studied the   1D UAV trajectory design  problem with the objective of maximizing    the minimal received   energy among  all ground nodes, subject to the maximum UAV speed constraints.
Different from previous works that only provided heuristic and locally optimal solutions, for the first time, we   presented  the globally  optimal 1D UAV trajectory solution  to the considered WPT problem, by equivalently decomposing any speed-constrained 1D UAV trajectory into a maximum-speed trajectory and a speed-free trajectory, together with the Lagrange dual method. In addition,   we have characterized the structure of  optimal trajectory solutions to the   WPT problem, i.e., an optimal trajectory can be described by a finite number of hovering points and the corresponding  hovering durations, while the UAV  always flies with the maximal speed among these hovering points. Moreover, we have derived the upper-bound for the  number of these hovering points.

The proposed optimal algorithm is based on exhaustive search, i.e.,  it leads to significant complexity.
 Future work will follow the optimal structure of the trajectory provided in this work to propose efficient trajectory designs and will extend the study to a 2D/3D system topology. 

\appendices
\section{Proof of Lemma~1}
This lemma can be proved by partitioning the whole time duration $T$  into a sufficiently large number  of time portions, each with a sufficiently small length such that during the portion the UAV speed is constant.  Denote the length of  $i$-th portion by $\tau_i, i =1,\cdots,I$ and we have $\sum\nolimits_{i=1,\cdots,I} \tau_i = T$. In addition, denote by $v_i$ the speed of  the UAV at  the $i$-th portion, i.e., $0 \le v_i \le V$. Hence, there are three cases at each portion: Case 1.    the UAV  hovers  at a given location, i.e.,  $  v_i =0$;  Case 2. the UAV    flies from $x_{i,{\rm{I}}}$ to  $x_{i,{\rm{F}}}$ with speed $ v_i =V$;  Case 3.  the UAV    flies from $x_{i,{\rm{I}}}$ to  $x_{i,{\rm{F}}}$ with speed $ 0<v_i < V$.

In the following,  we prove Lemma~\ref{le:Lemma1} by showing that   within each time portion the UAV trajectory satisfying the maximum speed constraint  is  equivalent to two trajectories as defined in the lemma.
The  $i$-th portion  of $\{x(t)\}$, the corresponding parts in $\{\bar x(t)\} $ and $\{\hat x(t)\}$   can be developed in the following way:
\begin{itemize}
\item Case 1: When  the UAV is hovering in the portion, just let the $\{\hat x(t)\}$   have the same hovering point and the same hovering time $\tau_i$. 
\item Case 2:  When the UAV flies from $x_{i,{\rm{I}}}$ to  $x_{i,{\rm{F}}}$  with the maximal speed, i.e., $ v_i = V$, just let trajectory  $\{\bar x(t)\} $  have the same trajectory as $\{  x(t)\} $  in this portion.
Hence, in Case 2   trajectory  $\{\hat x(t)\} $  not covers the interval between $x_{i,{\rm{I}}}$ and  $x_{i,{\rm{F}}}$. Hence, the   trajectory (in terms of not time but topology) of  $\{\hat x(t)\} $ is not continuous, i.e., there is not speed limit of the UAV in   $\{\hat x(t)\} $.
\item Case 3:  In this case,  the UAV flies   from $x_{i,{\rm{I}}}$ to  $x_{i,{\rm{F}}}$ with a speed lower than the maximal speed, i.e., $ 0<v_i < V$.
 The length of the portion is $\tau_i = \frac{{x_{i,{\rm{F}}}} - {x_{i,{\rm{I}}}}}{v_i}$.
 As shown in Fig.~\ref{Fig:forcase_3}, we can let the UAV fly with the maximal speed in $\{\bar x(t)\}$ which has  the time cost $\frac{{x_{i,{\rm{F}}}} - {x_{i,{\rm{I}}}}}{V}$. In addition, we let the UAV in $\{\hat x(t)\} $ use the remaining time, i.e., ${\frac{{{x_{i,{\rm{F}}}} - {x_{i,{\rm{I}}}}}}{{{v_i}}} - \frac{{{x_{i,{\rm{F}}}} - {x_{i,{\rm{I}}}}}}{V}}$, to fly from  $x_{i,{\rm{I}}}$ to  $x_{i,{\rm{F}}}$, while the corresponding speed can be calculated as $\frac{{{x_{i,{\rm{F}}}} - {x_{i,{\rm{I}}}}}}{{\frac{{{x_{i,{\rm{F}}}} - {x_{i,{\rm{I}}}}}}{{{v_i}}} - \frac{{{x_{i,{\rm{F}}}} - {x_{i,{\rm{I}}}}}}{V}}}= \frac{V v_i} {V-v_i}$.  Note that when $v_i$ becomes   significantly close to $V$ and therefore the corresponding UAV speed in $\{\hat x(t)\} $ is possible to be sufficient large, which confirms again no speed limit for the UAV in  $\{\hat x(t)\}$. As validated in~\eqref{eqn_lemma1}, in $i$-th portion the WPT performance of  $\{x(t)\} $ and the sum WPT performance of   $\{\bar x(t)\} $ and  $\{\hat x(t)\} $  are the same $\forall k = 1,\cdots,K$. 

\end{itemize}
So far, we have shown that  for each  portion  of $\{x(t)\}$, we can obtain the corresponding parts of $\{\bar x(t)\} $ and $\{\hat x(t)\}$  having the same WPT performance as the portion of $\{x(t)\}$. By repeating the above process for every portion  of $\{x(t)\}$,  $\{\bar x(t)\} $ and $\{\hat x(t)\}$  can be developed while satisfying Lemma~1.

 
\bibliographystyle{IEEEtran}

\bibliographystyle{IEEEtran}


\end{document}